\documentclass[submission,copyright,creativecommons, noderivs, noncommercial]{eptcs}
\providecommand{\event}{GASCom 2026} 

  \usepackage[T1]{fontenc}        

\usepackage{newunicodechar}
\newunicodechar{’}{'}

\usepackage{enumitem}

\usepackage{booktabs} 
\usepackage{adjustbox} 
\usepackage{cancel}
\usepackage{epigraph}
\usepackage{subcaption}
\setlist[enumerate, 1]{label=\arabic*.  }

\usepackage[table, dvipsnames]{xcolor} 
\usepackage{graphicx}
\usepackage{amsmath}
\usepackage{amssymb}
\usepackage{amsthm}
\usepackage{epstopdf}
\RequirePackage{doi}
\usepackage{hyperref}
\usepackage{forest}
\usepackage{tikz}
\usepackage{tikz-cd}
\usepackage{stmaryrd}
\usepackage{comment} 
\usepackage{nicefrac}
\usepackage{pgfplots}
\pgfplotsset{compat=1.16}

\definecolor{USred}{cmyk}{0,1.00,0.65,0.34}
\definecolor{darkblue}{rgb}{0.0627, 0.1539, 0.4510}
\definecolor{darkgreen}{rgb}{0.0627,0.4510,0.1539}

\definecolor{darkorange}{cmyk}{0,0.55,1,0.20}

\renewcommand{\emph}[1]{{\textcolor{USred}{\em #1}}}

\hypersetup{
  colorlinks,
  citecolor=USred,
  linkcolor=USred,
  urlcolor=USred}

\AtBeginDocument{

\def\SS{\mathbb{S}}

\DeclareMathOperator{\blob}{RecordCode}

\DeclareMathOperator{\rec}{rec}

\DeclareMathOperator{\mrec}{mrec}

\newcommand{\planted}{R_{\bullet}}

\setcounter{MaxMatrixCols}{20}
\setlength{\epigraphwidth}{0.48\textwidth}

\theoremstyle{definition}
\newtheorem{de}{Definition}
\theoremstyle{plain}
\newtheorem{thm}[de]{Theorem}

\newtheorem{lem}[de]{Lemma}

\newtheorem{p}[de]{Proposition}
\newtheorem{cor}[de]{Corollary}
\newtheorem{ex}[de]{Example}

\setlength{\epigraphwidth}{0.5\textwidth}
\theoremstyle{definition}

}

\title{The Genesis Sequence, Tree Records and Endofunctions.}
\author{Enrica Duchi
\institute{Université de Paris Cité, IRIF, Paris, France}
\and
Adrián Lillo
\institute{Universidad de Sevila, Sevilla, España}
\and
Pablo Puerto
\institute{Universidad de Sevila, Sevilla, España}
\and
Mercedes Rosas
\institute{Universidad de Sevila, Sevilla, España}
\and
Stefan Trandafir
\institute{Simon Fraser University, Vancouver, Canada}
}
\def\titlerunning{The Genesis Sequence, Tree Records and Endofunctions.}
\def\authorrunning{E. Duchi, A. Lillo, P. Puerto, M. Rosas \& S. Trandafir}

\begin{document}

\maketitle
\begin{abstract}
We present bijections connecting tree records, the girth of a connected endofunction, and the genesis sequence (the first sequence in OEIS). Using these, we derive generating functions for tree and forest record numbers in terms of Cayley’s tree function and give a new proof of Cayley’s forest formula.
\end{abstract}

\epigraph{\em
        This work is dedicated to Neil J. A. Sloane, renowned number-crusher, for his creation and tireless curation of the On-Line Encyclopedia of Integer Sequences.
}

\tikzset{
    record/.style={
        circle,
        draw=black,
        fill=black,
        inner sep=1.5pt,
        label={#1},
        solid
    },
    record/.default={},
    non-record/.style={
        circle,
        draw = black,
        fill = white,
        inner sep=1.5pt,
        solid
    },
    every label/.style={
            font=\tiny
        }
}


The present work is an extended abstract of \cite{LRT-bijective-visit}, building on earlier studies of records in rooted trees and forests \cite{LRT-GF-records, LRT-Weary}. Trees and forests are fundamental structures in discrete mathematics, and our work highlights an intriguing link between tree records, node heights, and endofunctions first suggested by the genesis sequence, the inaugural entry of the Online Encyclopedia of Integer Sequences. The observation that this sequence counts all records in rooted trees with $n$ nodes motivated combinatorial proofs of several identities relating tree records to the Cayley tree function. Space constraints prevent us from including results on the record code and a new proof of Cayley’s forest formula, these appear in \cite{LRT-bijective-visit}.

In the Numberphile podcast episode “The Number Collector” \cite{podcast_Sloan}, Neil Sloane recounts how he encountered the genesis sequence of the OEIS, the sequence  ``Normalized total height of all nodes in all rooted trees with $n$ labeled nodes",  during his PhD research, and  how it inspired the creation of the OEIS.
A short version of this story appears at the encyclopedia's wiki that we quote:
``The sequence database was begun by Neil J. A. Sloane  in early 1964 when he was a graduate student at Cornell University in Ithaca, NY. He had encountered a sequence of numbers while working on his dissertation, namely 1, 8, 78, 944,  $\ldots$
 (now entry \href{https://oeis.org/A000435}{A000435} in the OEIS \cite{oeis}), and was looking for a formula for the $n$-th term, in order to determine the rate of growth of the terms.
He noticed that although several books in the Cornell library contained sequences somewhat similar to this, this particular sequence was not mentioned. In order to keep track of the sequences in these books, NJAS started recording them on file cards, which he sorted into lexicographic order."
(From \href{https://oeis.org/wiki/Welcome#OEIS:_Brief_History}{Brief History of the OEIS}.)

    As all structures we consider are labeled, we simply use the terms tree and forest to refer to labeled trees and labeled forests.  
 Let $T$ be a rooted tree. The number of nodes in a tree or a forest (not counting $\circ$ whenever  it happens to appear as one of the labels) defines its \emph{order}.
A node $v$ of $T$ is a \emph{record} if it has the largest label along the unique path from $v$ to the root. Otherwise we say that $v$  is a non-record.  
The \emph{tree record number} $\planted(n,k)$ is defined as the number of rooted trees of order $n$ with precisely $k$ records. The tree record numbers are already present in the literature  in the context of queues
and traffic lights in the work of Haight  \cite{Haight}, \href{https://oeis.org/A259334}{A259334}.
The \emph{height} of a node of $T$ is the length of the path (number of edges) from it to the root. 
The \emph{total height} of $T$ is the sum of the heights of the nodes of $T$.

Write $[n]$ for $\{1, 2, \ldots, n\}$ and $[n]_0$ for $\{0, 1, \ldots, n\}$.
An \emph{endofunction} on $[n]$ is a function $f : [n] \to [n]$ for some  $n$. Its \emph{functional digraph} is  the directed graph with nodes $[n]$, and one directed edge $(i,f(i))$ for each $i$ in $[n]$.
Given a rooted tree $T$, we define the \emph{parent map} of $T$ as the map that sends each non-root vertex to its parent, and the root to itself.  Observe that this is a natural way of embedding the set of rooted trees in the set of endofunctions.
An endofunction  is said to be \emph{connected} if its functional graph is connected after ignoring orientations. The \emph{girth} of a connected endofunction is the number of edges in the unique cycle of its functional graph.

\subsection*{The height of a record and the girth of an endofunction}

Through the lens of Joyal’s bijection \cite{Joyal}, we reveal an elegant relationship between record numbers and connected endofunctions that constitutes the starting point of our journey. Fix two nonnegative integers $k \le n$.

\begin{thm} 
\label{thm:connected endofunctions}
The number of rooted trees of order $n$ with a distinguished record at height $k-1$ coincides with the number of connected endofunctions on $[n]$ of girth $k$.
\end{thm}

\begin{proof}[Bijective proof]  Let $T$ be a rooted   tree with a distinguished record $v$,  root $r$. Let $p$ be its parent map.
We send the pair $(T, v)$ to the endofunction 
that sends the root $r$ to the distinguished record $v$, and any other vertex to its parent in $T$.  See Figure \ref{fig:endomorphism bijection}.
Note that $f$ is a connected endofunction, as
$T$ is a spanning tree of the graph of $f$.
 
The inverse map of  $(v, T) \mapsto f$ can be described as follows. 
Let $f$ be a connected endofunction of $[n]$, and let $\sigma =(x_1 \ x_2 \ \cdots \ x_m)$ be the only cycle of $f$. First, we need to determine the vertices $v$ and $r$. These are given by
$
v = \max_i x_i,$ and $r = \sigma^{-1}(v).
$
 We reconstruct the tree $T$ by first removing the edge that connects $v$ with $r$ in the graph of $f$, and then rooting the resulting tree at $r$. Vertex $v$ is a record of $T$ as the vertices in the path from $v$ towards the root are the elements of $\sigma$, and we have defined $v$ as the maximum of these vertices. Consequently, the height of $v$ is exactly one less than the girth of $f$.

\qedhere
\end{proof}

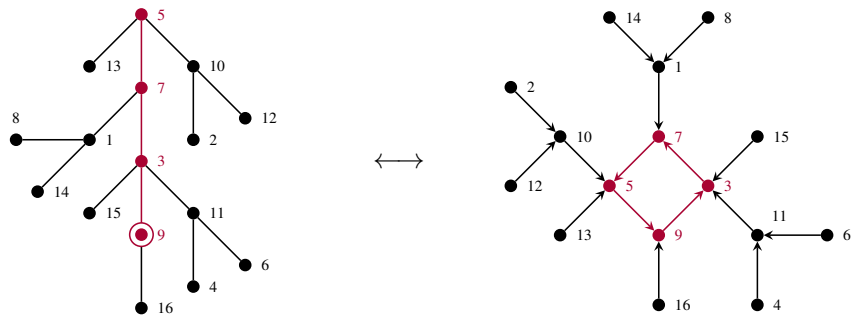
\begin{figure}[h!]
    \centering

\resizebox{0.75\textwidth}{!}{$
  \vcenter{\hbox{\usetikzlibrary{positioning}
\tikzset{
  circ/.style={
    circle, draw, semithick, minimum size=1.5mm, inner sep=0pt, fill=black,
    }
}

\begin{tikzpicture}[
    arr/.style={semithick}
]
    \newcommand{\len}{1}

    \node[circ, label={right:$\textcolor{USred}{5}$}, fill=USred,USred] (s1) at (0,0) {};
    \node[circ, label={right:$\textcolor{USred}{7}$}, fill=USred,USred] (s2) at ($(s1) + (-90:\len)$) {};
    \node[circ, label={right:$\textcolor{USred}{3}$}, fill=USred,USred] (s3) at ($(s2) + (-90:\len)$) {};
    \node[circ, label={right:$\textcolor{USred}{9}$}, fill=USred,USred] (s4) at ($(s3) + (-90:\len)$) {};

    \node[circ, label={right:$16$}] (s5) at ($(s4) + (-90:\len)$) {};
    
    \node[circ, label={right:$11$}] (s6) at ($(s3) + (-45:\len)$) {};
    \node[circ, label={right:$6$}] (s7) at ($(s6) + (-45:\len)$) {};
    \node[circ, label={right:$4$}] (s8) at ($(s6) + (-90:\len)$) {};
    \node[circ, label={right:$15$}, label distance=0] (s16) at ($(s3) + (-135:\len)$) {};

    \node[circ, label={right:$1$}] (s10) at ($(s2) + (-135:\len)$) {};
    \node[circ, label={right:$14$}] (s11) at ($(s10) + (-135:\len)$) {};
    \node[circ, label={above:$8$}] (s12) at ($(s10) + (-180:\len)$) {};
    \node[circ, label={right:$2$}] (s9) at ($(s2) + (-45:\len)$) {};

    \node[circ, label={right:$13$}] (s13) at ($(s1) + (-135:\len)$) {};
    \node[circ, label={right:$10$}] (s14) at ($(s1) + (-45:\len)$) {};
    \node[circ, label={right:$12$}] (s15) at ($(s14) + (-45:\len)$) {};

    \begin{scope}[USred]
    \draw[arr] (s2) -- (s1);
    \draw[arr] (s3) -- (s2);
    \draw[arr] (s4) -- (s3);
    \end{scope}

    \draw[arr] (s5) -- (s4);
    
    \draw[arr] (s6) -- (s3);
    \draw[arr] (s7) -- (s6);
    \draw[arr] (s8) -- (s6);
    \draw[arr] (s16) -- (s3);
    
    \draw[arr] (s10) -- (s2);
    \draw[arr] (s11) -- (s10);
    \draw[arr] (s12) -- (s10);
    
    \draw[arr] (s13) -- (s1);
    \draw[arr] (s14) -- (s1);
    \draw[arr] (s15) -- (s14);
    \draw[arr] (s9) -- (s14);

    \draw[draw=USred, semithick, fill=white] (s4) circle [radius=1.6mm];
    \node[circ, fill=USred, draw=USred] (s4bis) at (s4) {};

\end{tikzpicture}}}%
  \hspace{1cm}%
  \longleftrightarrow%
  \hspace{1cm}%
  \vcenter{\hbox{\begin{tikzpicture}[
    circ/.style={circle, draw, semithick, minimum size=1.5mm, inner sep=0pt, fill=black},
    arr/.style={-stealth, semithick}
]
    \newcommand{\len}{0.95}
    \begin{scope}[USred]    
    \node[circ,  label={right:$\scriptsize 7$}, fill=USred,USred] (c1) at (0, 0) {};
    \node[circ, label={right:$\scriptsize 5$}, fill=USred,USred] (c2) at ($(c1) + (-135:\len)$) {};
    \node[circ,  label={right:$\scriptsize 9$}, fill=USred,USred] (c3) at ($(c2) + (-45:\len)$) {};
    \node[circ, label={right:$\scriptsize 3$}, fill=USred,USred] (c4) at ($(c3) + (45:\len)$) {};
    \end{scope}

    \node[circ, label={right:$\scriptsize 1$}] (c5) at ($(c1) + (90:\len)$) {};
    \node[circ, label={right:$\scriptsize 8$},] (c6) at ($(c5) + (45:\len)$) {};
    \node[circ, label={right:$\scriptsize 14$},] (c7) at ($(c5) + (135:\len)$) {};

    \node[circ, label={right:$\scriptsize 16$},] (c8) at ($(c3) + (-90:\len)$) {};

    \node[circ, label={right:$\scriptsize 15$},] (c9) at ($(c4) + (45:\len)$) {};
    \node[circ, label={above right:$\scriptsize 11$},] (c10) at ($(c4) + (-45:\len)$) {};
    \node[circ, label={right:$\scriptsize 6$},] (c11) at ($(c10) + (0:\len)$) {};
    \node[circ, label={right:$\scriptsize 4$},] (c12) at ($(c10) + (-90:\len)$) {};
    
    \node[circ, label={right:$\scriptsize 10$},] (c13) at ($(c2) + (135:\len)$) {};
    \node[circ, label={right:$\scriptsize 13$},] (c14) at ($(c2) + (-135:\len)$) {};
    \node[circ, label={right:$\scriptsize 2$},] (c15) at ($(c13) + (135:\len)$) {};
    \node[circ, label={right:$\scriptsize 12$},] (c16) at ($(c13) + (-135:\len)$) {};

    \begin{scope}[USred]
    \draw[arr] (c1) -- (c2);
    \draw[arr] (c2) -- (c3);
    \draw[arr] (c3) -- (c4);
    \draw[arr] (c4) -- (c1);
    \end{scope}

    \draw[arr] (c5) -- (c1);
    \draw[arr] (c6) -- (c5);
    \draw[arr] (c7) -- (c5);
    
    \draw[arr] (c16) -- (c13);
    \draw[arr] (c14) -- (c2);
    \draw[arr] (c15) -- (c13);
    \draw[arr] (c13) -- (c2);

    \draw[arr] (c8) -- (c3);

    \draw[arr] (c9) -- (c4);
    \draw[arr] (c10) -- (c4);
    \draw[arr] (c11) -- (c10);
    \draw[arr] (c12) -- (c10);
    
\end{tikzpicture}}}
  $}
\caption{The bijection of Theorem \ref{thm:connected endofunctions}. The distinguished record $v = 9$ is circled. The path from $v$ to the root (on the left) and the only cycle in the resulting endofunction (on the right) are highlighted.}
    \label{fig:endomorphism bijection}
\end{figure}

\begin{cor} 
\label{cor:connected endofunctions}
The total number of records in all rooted trees with $n$ nodes is equal to the number of connected endofunctions on $[n]$. 
\end{cor}

The sequence counting connected endofunctions has been studied in depth, see \cite{Flajolet}, or the entries of the OEIS \href{https://oeis.org/A001865}{A001865} and  \href{https://oeis.org/A201685}{A20168}. 

\subsection*{Some consequences}

Let $f$ be an endofunction on $[n]$.   The set of elements of the form $f^k(i)$ for some $k\ge0$ is referred to as the \emph{orbit} of  $i$, for any fixed element $i$ of $[n]$. If the endofunction $f$ is connected we refer to the \emph{cycle} of $f$ as the cyclic permutation obtained by restricting $f$ to the elements of its unique cycle.
Finally, an element $i \in [n]$ is said to be a \emph{record} if $ i$ is larger than or equal to all the elements of its orbit. Fix three nonnegative integers $k,m \le n$.

\begin{thm}
\label{Thm:endo_girth}
    The number of connected endofunctions on $[n]$ of girth $m$ and at least $k + 1$ records coincides with the number of connected endofunctions on $[n]$ of girth $m + 1$ and at least $k$ records.  
\end{thm}

\begin{proof}[Bijective proof] 
    Let $f$ be a connected endofunction of girth $m$ and at least $k+1$ records. Let $v$ be the $k$-th greatest record of $f$.  
    Observe that there is exactly one record on the cycle $\sigma$ of $f$: the smallest record of the endofunction, and that it is smaller than the $k$-th greatest record. Therefore, $v$ does not belong to $\sigma$.
    
Let $w$ be the first element in the orbit of $v$ under $f$ such that $f(w)$ belongs to $\sigma$.
Denote $x = \sigma^{-1}(f(w))$.
We define an endofunction $g$ by 
$
g(z) = f(z) \quad \text{if } z \neq x,$ and $ g(x) = w.
$
It is not hard to verify that $g$ is a connected endofunction of girth $m+1$ and its $k$ greatest records coincide with those of $f$. Reversing the bijection is straightforward by first identifying $v$ and $w$.
\end{proof}

\begin{figure}
\[
\centering
\vcenter{\hbox{
\begin{tikzpicture}[
    circ/.style={circle, draw, semithick, minimum size=1.5mm, inner sep=0pt, fill=black},
    rec/.style={circle, draw=USred,  minimum size=1.5mm, inner sep=0pt, fill=USred},
    arr/.style={-stealth, semithick}
]
    \def\sepFirstLevel{0.8}
    \def\len{0.8}
    \node[rec, label={left:3}] (3) at (0, 0) {};
    \node[rec, label={right:$4$}] (4) at ($(3) + (-50:\len)$) {};
    \node[rec, label={right:$5$}, label={left:\textcolor{USred}{$v$}}, label={above:\textcolor{USred}{$w$}}] (5) at ($(3) + (-130:\len)$) {};
    \node[rec, label={below:$6$}] (6) at ($(5) + (-120:0.9*\len)$) {};
    \node[circ, label={below:$2$}] (2) at ($(5) + (-60:0.9*\len)$) {};
    
    \node[circ, label={below:$1$}] (1) at ($(4) + (-120:0.9*\len)$) {};
    \node[rec, label={below:$7$}] (7) at ($(4) + (-60:0.9*\len)$) {};

\begin{scope}[-stealth]    
   \draw (3) edge[-stealth, out=45, in=135, looseness=12] (3); 
    \draw (6) -- (5);
    \draw (2) -- (5);
    \draw (5) -- (3);
    \draw (4) -- (3);
    \draw (1) -- (4);
    \draw (7) -- (4);
\end{scope}
\end{tikzpicture}}}
\qquad
\vcenter{\hbox{
\begin{tikzpicture}[
    circ/.style={circle, draw, semithick, minimum size=1.5mm, inner sep=0pt, fill=black},
    rec/.style={circle, draw=USred,  minimum size=1.5mm, inner sep=0pt, fill=USred},
    arr/.style={-stealth, semithick}
]
    \def\sepFirstLevel{0.8}
\def\len{0.8}
    \node[circ, label={above:3}] (3) at (0, 0) {};
    \node[circ, label={above:$4$}] (4) at ($(3) + (0:\len)$) {};
    \node[rec, label={above:$5$}] (5) at ($(3) + (-180:\len)$) {};
    \node[rec, label={above:$7$}] (7) at ($(4) + (30:\len)$) {};
    \node[circ, label={above:$1$}] (1) at ($(4) + (-30:\len)$) {};
    \node[rec, label={above:$6$}, label={left:\textcolor{USred}{$v$}},
    label={right:\textcolor{USred}{$w$}}] (6) at ($(5) + (150:\len)$) {};
    \node[circ, label={above:$2$}] (2) at ($(5) + (210:\len)$) {};

    \begin{scope}[-stealth]
    \draw (6) -- (5);
    \draw (2) -- (5);
    \draw (4) -- (3);
    \draw (5) edge[bend right] (3);
    \draw (3) edge[bend right] (5);
    \draw (1) -- (4);
    \draw (7) -- (4); 
    \end{scope}
\end{tikzpicture}
}}
\qquad
\vcenter{\hbox{
\begin{tikzpicture}[
    circ/.style={circle, draw,  minimum size=1.5mm, inner sep=0pt, fill=black},
    rec/.style={circle, draw=USred,  minimum size=1.5mm, inner sep=0pt, fill=USred},
    arr/.style={-stealth}
]
    \def\len{0.8}
    \node[circ, label={below:$\scriptsize 3$},] (6) at (0, 0) {};
    \node[circ, label={below:$\scriptsize 5$},] (5) at ($(6) + (180:\len)$) {};
    \node[rec, label={left:$\scriptsize 6$},] (3) at ($(6) + (120:\len)$) {};
    \node[circ, label={above:$4$}, label={[yshift=2pt, xshift=2pt]below left:\textcolor{USred}{$w$}}] (4) at ($(6) + (-30:\len)$) {};
    \node[rec, label={above:$7$}, label={right:\textcolor{USred}{$v$}}] (7) at ($(4) + (0:\len)$) {};
    \node[circ, label={above:$1$}] (1) at ($(4) + (-60:\len)$) {};
    \node[circ, label={below:$2$}] (2) at ($(5) + (-150:\len)$) {};

    \begin{scope}[-stealth]
        \draw (5) -- (6);
        \draw (3) -- (5);
        \draw (6) -- (3);
        \draw (4) -- (6);
        \draw (1) -- (4);
        \draw (7) -- (4);
        \draw (2) -- (5);
    \end{scope}
\end{tikzpicture}}}
\qquad
\vcenter{\hbox{
\begin{tikzpicture}[
    circ/.style={circle, draw, semithick, minimum size=1.5mm, inner sep=0pt, fill=black},
    rec/.style={circle, draw=USred,  minimum size=1.5mm, inner sep=0pt, fill=USred},
    arr/.style={-stealth}
]
    \newcommand{\len}{0.8}
    \begin{scope}    
    \node[circ,  label={right:$\scriptsize 3$}] (c1) at (0, 0) {};
    \node[rec, label={right:$\scriptsize 6$}] (c2) at ($(c1) + (-135:\len)$) {};
    \node[circ,  label={right:$\scriptsize 5$}] (c3) at ($(c2) + (-45:\len)$) {};
    \node[circ, label={right:$\scriptsize 4$}] (c4) at ($(c3) + (45:\len)$) {};
    \end{scope}

    \node[rec, label={right:$\scriptsize 7$},] (15) at ($(c4) + (45:\len)$) {};
    \node[circ, label={right:$\scriptsize 1$},] (11) at ($(c4) + (-45:\len)$) {};
    
    \node[circ, label={right:$\scriptsize 2$},] (10) at ($(c3) + (-90:\len)$) {};

    \begin{scope}
    \draw[arr] (c1) -- (c2);
    \draw[arr] (c2) -- (c3);
    \draw[arr] (c3) -- (c4);
    \draw[arr] (c4) -- (c1);
    \end{scope}

    \draw[arr] (10) -- (c3);
    \draw[arr] (11) -- (c4);
    \draw[arr] (15) -- (c4);

\end{tikzpicture}
}}
\]
    \caption{The repeated application of the bijection in Theorem \ref{Thm:endo_girth}. The records are highlighted. Note that the leftmost endofunction has at least $4$ records and girth $1$, and the rightmost endofunction has at least $1$ record and girth $4$.}
    \label{fig:placeholder}
\end{figure}
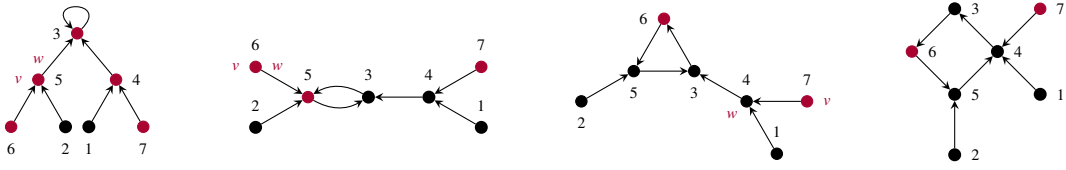

\begin{cor}
\label{pablos_lemma}
The number of connected endofunctions on $[n]$ of girth $m$ and at least $k$ records only depends on $m+k$. In particular, the number of rooted trees with $n$ nodes and at least $k$ records coincides with the number of connected endofunctions on $[n]$ of girth $k$.
\end{cor}
\begin{proof}[Bijective proof]
Fix two positive integers $k$ and $m$ and let $E_{m,k}$ denote the set of endofunctions of $[n]$ with girth $m$ and at least $k$ records. Applying repeatedly the bijection of Theorem \ref{Thm:endo_girth}, we obtain the following 
bijections:
\[
E_{m, k} \longleftrightarrow E_{m +1, k -1} \longleftrightarrow E_{m+2, k-2} \longleftrightarrow \dots \longleftrightarrow E_{m + k - 1, 1}.
\]
Therefore, the cardinality of $E_{m, k}$ only depends on $m + k$. Thus, with $m=1$, $
|E_{1, k}| = |E_{k, 1}|.$
The result follows by noticing that endofunctions with girth 1 are parent maps of trees, and that any endofunction has at least one record.
\end{proof}

We conclude that the number of rooted trees of order $n$ with at least $k$ records coincides with the number of tree cycles of order $n$ of $k$ rooted trees.
 The number of trees of order $n$ with $k$ records such that the subgraph induced by its records is connected coincides with the number of cycles of order $n$ of $k$ unrooted trees.
 The set of rooted trees of order $n$ with a distinguished record at height $k-1$ is in bijection with the set of trees of order $n$ with at least $k$ records. 
    The set of rooted trees of order $n$ and with a distinguished record $v$ at height $m$ and with at least $k$ records greater than $v$ is in bijection with the set of endofunctions on $[n]$ of girth $m+k+1$.

\subsection*{The tree record formula}
We provide a combinatorial proof for the elegant closed formula for the tree record numbers of \cite{LRT-GF-records}.

\begin{cor}[Theorem 4.1 of \cite{LRT-GF-records}]
\label{prop:R_planted}
  The number of  rooted   trees labeled with  $[n]$ and  with $k$ records obeys the equation
    \[
    \planted(n,k) = 
       k \,(n-1)\cdots
    (n-k+1) \, n^{n-k-1}.
     \]
\end{cor}

\begin{proof}
The number of endofunctions of girth $k$ is 
$
n(n-1)\dots(n-k+1) \ n^{n-k-1}
$
as $kn^{n-k-1}$ counts the number of  labeled rooted forests  with a fixed $k$-set of roots (Cayley's forest formula), and $n(n-1)\dots(n-k+1)/k$ counts the number of ways to choose the $k$-set of roots of the trees and  order it cyclically. Finally, Corollary \ref{pablos_lemma} allows us to deduce the formula for the record numbers by subtracting  the number of endofunctions of girth $k+1$ from  those of girth $k$.
\end{proof}

\subsection*{The genesis sequence of the On-Line Encyclopedia of Integer Sequences}

\label{se:genesis}

The \emph{genesis sequence} $\Gamma = (\gamma_k)_{k}$ of the OEIS  is defined by setting $\gamma_k$ to be the sum of the height of all nodes of  all rooted trees of order $k$, divided by $k$, OEIS \href{https://oeis.org/A000435}{A000435}. 
Let $T$ be a rooted tree. A \emph{catalyst} for $T$ to be a pair $(u, v)$ of distinct vertices of $T$ such that $v$ is an ancestor of $u$, and denote by $C(T)$ the set of catalysts for $T$.  A closely related, but different, notion of catalyst is introduced in \cite{briand2025nonintersectingpathsdeterminantdistance}. The number of catalysts for $T$ is equal to the sum of the heights of all vertices in $T$.

\begin{thm}
\label{thm:catalysts}
The  total number of non-root records in all rooted trees labeled with $[n]$ is equal to $\gamma_n$, the normalized total height of all nodes in all rooted trees with $n$ labeled nodes.
\end{thm}

\begin{proof}[Bijective proof]
    Let $P(n)$ denote the set of pairs $(T, r)$ where $T$ is a rooted tree labeled with $[n]$ and $r$ is a record of $T$ different from the root. 
    We exhibit an explicit bijection between the set $[n]\times P(n)$ and the set 
    $C(n) = \bigsqcup_T C(T).$
    Note that the existence of such a bijection implies the result. 
    Let $(v, T, r)$ be an element of $[n] \times P(n)$.
     We shall construct a tree $T'$ and a vertex $x$ such that $(x, v)$ is a catalyst for $T'$. 
    First, we compute the endofunction $f$ corresponding with the pair $(T, r)$ under the bijection of Theorem \ref{thm:connected endofunctions}. Let $\sigma$ denote the cycle of $f$, and
    let $x$ be the first element in the orbit of $v$ under $f$ that belongs to the cycle $\sigma$. Define an endofunction $g$ by
    $
    g(z) = f(z)$ if $z\neq \sigma^{-1}(x),$ and $g(\sigma^{-1}(x)) = v,$ otherwise. 
    Let $\pi$ denote the cycle of $g$ and let $w = \pi^{-1}(x)$. 
    Lastly, define $T'$ to be the tree with parent map $p$, where $p$ is given by 
    $p(z) = g(z)$ if  $z\neq w, p(w)=w. $
    Then root $T'$  at $w$. It is simple to check that $(x, v) \in C(T')$, and that all steps are reversible.    
\end{proof}

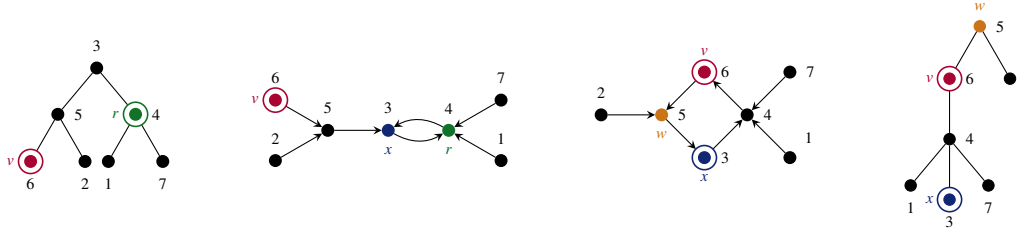
\begin{figure}
    \centering
    $\centering
\vcenter{\hbox{
\begin{tikzpicture}[
    circ/.style={circle, draw, semithick, minimum size=1.5mm, inner sep=0pt, fill=black},
    arr/.style={-stealth, semithick}
]
    \def\sepFirstLevel{0.8}
    \def\len{0.8}
    \node[circ, label={above:3}] (3) at (0, 0) {};
    \node[circ, label={right:$4$}] (4) at ($(3) + (-50:\len)$) {};
    \node[circ, label={right:$5$}] (5) at ($(3) + (-130:\len)$) {};
    \node[circ, label={below:$6$}] (6) at ($(5) + (-120:0.9*\len)$) {};
    \node[circ, label={below:$2$}] (2) at ($(5) + (-60:0.9*\len)$) {};
    
    \node[circ, label={below:$1$}] (1) at ($(4) + (-120:0.9*\len)$) {};
    \node[circ, label={below:$7$}] (7) at ($(4) + (-60:0.9*\len)$) {};

    \draw (6) -- (5);
    \draw (2) -- (5);
    \draw (5) -- (3);
    \draw (4) -- (3);
    \draw (1) -- (4);
    \draw (7) -- (4);
    
    \draw[draw=USred, semithick, fill=white] (6) circle [radius=1.6mm];
    \node[circ, fill=USred, draw=USred,
    label={left:$\textcolor{USred}{v}$}
    ] (6bis) at (6) {};
    
    \draw[draw=darkgreen, semithick, fill=white] (4) circle [radius=1.6mm];
    \node[circ, fill=darkgreen, draw=darkgreen, label={left:$\textcolor{darkgreen}{r}$}] (4bis) at (4) {};
\end{tikzpicture}}}
\qquad
\vcenter{\hbox{
\begin{tikzpicture}[
    circ/.style={circle, draw, semithick, minimum size=1.5mm, inner sep=0pt, fill=black},
    arr/.style={-stealth, semithick}
]
    \def\sepFirstLevel{0.8}
\def\len{0.8}
    \node[circ, fill=darkblue, draw=darkblue, label={above:3}, label={below:$\textcolor{darkblue}{x}$}] (3) at (0, 0) {};
    \node[circ, draw=darkgreen, fill=darkgreen, label={above:$4$}, label={below:$\textcolor{darkgreen}{r}$}] (4) at ($(3) + (0:\len)$) {};
    \node[circ, label={above:$5$}] (5) at ($(3) + (-180:\len)$) {};
    \node[circ, label={above:$7$}] (7) at ($(4) + (30:\len)$) {};
    \node[circ, label={above:$1$}] (1) at ($(4) + (-30:\len)$) {};
    \node[circ, label={above:$6$}, label={left:$\textcolor{USred}{v}$}] (6) at ($(5) + (150:\len)$) {};
    \node[circ, label={above:$2$}] (2) at ($(5) + (210:\len)$) {};
    %
    

    \begin{scope}[-stealth]
    \draw (6) -- (5);
    \draw (2) -- (5);
    \draw (5) -- (3);
    \draw (4) edge[bend right] (3);
    \draw (3) edge[bend right] (4);
    \draw (1) -- (4);
    \draw (7) -- (4); 
    \end{scope}
    \draw[draw=USred, semithick, fill=white] (6) circle [radius=1.6mm];
    \node[circ, fill=USred, draw=USred] (6bis) at (6) {};
\end{tikzpicture}
}}
\qquad
\vcenter{\hbox{
\begin{tikzpicture}[
    circ/.style={circle, draw, semithick, minimum size=1.5mm, inner sep=0pt, fill=black},
    arr/.style={-stealth}
]
    \newcommand{\len}{0.8}
    \begin{scope}    
    \node[circ,  label={right:$\scriptsize 6$}, label={above:$\textcolor{USred}{v}$}] (c1) at (0, 0) {};
    \node[circ, draw=darkorange, fill=darkorange, label={right:$\scriptsize 5$}, label={below:$\textcolor{darkorange}{w}$}] (c2) at ($(c1) + (-135:\len)$) {};
    \node[circ,  label={right:$\scriptsize 3$}] (c3) at ($(c2) + (-45:\len)$) {};
    \node[circ, label={right:$\scriptsize 4$}] (c4) at ($(c3) + (45:\len)$) {};
    \end{scope}

    \node[circ, label={right:$\scriptsize 7$},] (15) at ($(c4) + (45:\len)$) {};
    \node[circ, label={above right:$\scriptsize 1$},] (11) at ($(c4) + (-45:\len)$) {};
    
    \node[circ, label={above:$\scriptsize 2$},] (10) at ($(c2) + (180:\len)$) {};

    \begin{scope}
    \draw[arr] (c1) -- (c2);
    \draw[arr] (c3) -- (c4);
    \end{scope}
    \draw[arr] (c4) -- ($(c1)+(-45:1.6mm)$);
    \draw[arr] (c2) -- ($(c3)+(135:1.6mm)$);

    \draw[arr] (10) -- (c2);
    \draw[arr] (11) -- (c4);
    \draw[arr] (15) -- (c4);
    
    \draw[draw=USred, semithick, fill=white] (c1) circle [radius=1.6mm];
    \node[circ, fill=USred, draw=USred] (c1bis) at (c1) {};
    
    \draw[draw=darkblue, semithick, fill=white] (c3) circle [radius=1.6mm];
    \node[circ, fill=darkblue, draw=darkblue, label={below:$\textcolor{darkblue}{x}$}] (c1bis) at (c3) {};
    
\end{tikzpicture}
}}
\qquad
\vcenter{\hbox{
\begin{tikzpicture}[
    circ/.style={circle, draw,  minimum size=1.5mm, inner sep=0pt, fill=black},
    arr/.style={-stealth}
]
    \def\len{0.8}
    \node[circ, draw=darkorange, fill=darkorange, label={right:$\scriptsize 5$}, label={above:$\textcolor{darkorange}{w}$}] (5) at (0, 0) {};
    \node[circ, label={right:$\scriptsize 2$},] (2) at ($(5) + (-60:\len)$) {};
    \node[circ, label={right:$\scriptsize 6$},] (6) at ($(5) + (-120:\len)$) {};
    \node[circ, label={right:$\scriptsize 4$},] (4) at ($(6) + (-90:\len)$) {};
    
    \node[circ, label={below:$\scriptsize 1$},] (1) at ($(4) + (-130:\len)$) {};
    \node[circ, label={below:$\scriptsize 3$},] (3) at ($(4) + (-90:\len)$) {};
    \node[circ, label={below:$\scriptsize 7$},] (7) at ($(4) + (-50:\len)$) {};

    \draw (1) -- (4);
    \draw (3) -- (4);
    \draw (7) -- (4);
    \draw (4) -- (6);
    \draw (6) -- (5);
    \draw (2) -- (5);
    
    \draw[draw=darkblue, semithick, fill=white] (3) circle [radius=1.6mm];
    \node[circ, fill=darkblue, draw=darkblue, label={left:$\textcolor{darkblue}{x}$}] (3bis) at (3) {};
    
    \draw[draw=USred, semithick, fill=white] (6) circle [radius=1.6mm];
    \node[circ, fill=USred, draw=USred,label={left:$\textcolor{USred}{v}$}] (6bis) at (6) {};
\end{tikzpicture}}}$
    \caption{The bijection $[n]\times R(n)\to C(n)$ from Theorem~\ref{thm:catalysts}. The left-most tree is rooted at $3$, the selected vertex is $v=6$ (crimson), and the selected nonroot record is $r=4$ (green). The second graph is obtained via Theorem~\ref{thm:connected endofunctions}. Finally, $w=\pi^{-1}(x)$ (yellow) is the preimage of $x$ on the unique cycle of the third graph and the root of the tree on the right.}
    \label{fig:first_sequence_bijection}
\end{figure}

\subsection*{Tree record generating functions}

 Fix a nonnegative integer $k$. 
Let $\mathcal{R}_k(z)$ be the generating function for rooted trees with precisely $k$ records, and let $\mathcal{R}_{\ge k}(z)$ be the generating function for rooted trees at least $k$ records:
\begin{align*}
\mathcal{R}_k(z) = \sum_{n\ge 0} \planted(n,k) \ z^n,
&&\mathcal{R}_{\ge k}(z) = \sum_{l\ge k}\mathcal{R}_l(z).
\end{align*}

\begin{cor} The generating function for trees with at least $k$ records is expressed in terms of the Cayley tree function as
\[
\mathcal{R}_{\ge k}(z)  = \frac{\mathcal{T}^k(z)}{k}.
\]

\end{cor}

In Proposition 3.8 of \cite{LRT-GF-records}, it was shown that the generating function for rooted trees with precisely $k$ records obeys the elegant equation
$
\mathcal{R}_k = \frac{\mathcal{T}^k(z)}{k} - \frac{\mathcal{T}^{k+1}(z)}{k+1}
$
Indeed, for  $k=1$, the equation  becomes the familiar expression for the generating function $\mathcal{U}(z)$ for unrooted trees:
$
\mathcal{U}(z) =  \mathcal{T}(z) - \frac{\mathcal{T}^{2}(z)}{2}
$
 (Observe that $\mathcal{U}(z)= \mathcal{R}_1(z)$ because
 trees rooted at their largest nodes are equinumerous with unrooted trees).  A beautiful proof of this equation is given by Pierre Leroux's dissymmetry lemma \cite{Leroux_dissym, Leroux_dissym_2, gessel2023goodhuntingsproblemcounting}.

\begin{lem}[A generalization of Pierre Leroux's dissimmetry lemma]
\label{cor:gf_krecords}
\begin{align}
\label{Eq:Leroux_generalization}
 \frac{\mathcal{T}^k(z)}{k} =
 \mathcal{R}_k + \frac{\mathcal{T}^{k+1}(z)}{k+1}. 
\end{align}
\end{lem}
\begin{proof}
There is a bijection between the set of endofunctions on $[n]$ of girth $k$ and the union of the set of trees of order $n$ with $k$ records and the set of endofunctions on $[n]$ of girth $k+1$.
Given an endofunction of girth $k$, the bijection of Corollary \ref{pablos_lemma} gives a tree $T$ with at least $k$ records. Either $T$ is in the set of trees of order $n$ with $k$ records, or it has at least $k+1$ records. In the latter case, the inverse of the same bijection gives an endofunction of girth $k+1$.

\end{proof}

\subsection*{Differential equations and the record generating function}

We provide a combinatorial interpretation for the derivative  $\mathcal{R}'_k(z)$ of the series 
$\mathcal{R}_k(z)$. This series satisfies
$
\mathcal{R}'_k(z)=\frac{\mathcal{T}^k(z)}{z}.
$
Indeed, $\mathcal{R}'_1(z) = \frac{\mathcal{T}(z)}{z}$, as can be  checked directly by taking the derivative with respect to $z$ of $\mathcal{R}_1(z)$, the generating function for unrooted trees, and the expression of $\mathcal{R}_k(z)$ follows directly from Lemma \ref{cor:gf_krecords} by differentiating with respect to $z$:
$
\mathcal{R}'_k(z)  = \mathcal{T}(z) \mathcal{R}_{k-1}'(z). 
$

Our aim is thus to give a combinatorial proof this equation. Let us consider the class of rooted Cayley trees $\mathcal{R}^*_k$ with $k$ records, one of which is a marked virtual node, i.e. a node with no label that does not contribute to the size, and such that the subtree of the virtual node does not contain any record or, in other terms, the path going from each record to the root does not cross the virtual node (except for the record corresponding to the virtual node), then we have the following lemma.
\begin{lem}
The series $\mathcal{R}'_k(z)$ is the generating function of the class $\mathcal{R}^*_k$.
\label{lem:LemVirtNode}
\end{lem}
\begin{figure}[h!]
    \begin{center}
        \resizebox{0.75\textwidth}{!}{$%
        \vcenter{\hbox{\usetikzlibrary{positioning}
\tikzset{
  circ/.style={
    circle, draw, semithick, minimum size=1.5mm, inner sep=0pt, fill=black,
    }
}

\begin{tikzpicture}[
    arr/.style={semithick}
]
    \newcommand{\len}{1}

    \node[circ, label={right:$14$}] (s1) at (0,0) {};
    \node[circ, label={right:$2$}] (s11) at ($(s1) + (-135:\len)$) {};
    \node[circ, label={right:$5$}] (s12) at ($(s1) + (-90:\len)$) {};
    \node[circ, label={right:$6$}] (s13) at ($(s1) + (-45:\len)$) {};
    \node[circ, label={right:$8$}] (s111) at ($(s11) + (-135:\len)$) {};
    \node[circ, label={right:$11$}] (s112) at ($(s11) + (-90:\len)$) {};
    \node[circ, label={right:$13$}] (s113) at ($(s11) + (-40:\len)$) {};
    \node[circ, label={right:$7$}] (s1121) at ($(s112) + (-135:\len)$) {};
    \node[circ, label={right:$18$}] (s1122) at ($(s112) + (-45:\len)$) {};
    \node[circ, label={right:$16$}] (s11221) at ($(s1122) + (-135:\len)$) {};
    \node[circ, label={right:$9$}] (s1131) at ($(s113) + (-45:\len)$) {};
    \node[circ, label={right:$10$}] (s11311) at ($(s1131) + (-115:\len)$) {};
    \node[circ, label={right:$19$}] (s11312) at ($(s1131) + (-45:\len)$) {};
     \node[circ, label={right:$15$}] (s113121) at ($(s11312) + (-135:\len)$) {};
     \node[circ, label={right:$17$}] (s113122) at ($(s11312) + (-90:\len)$) {};
     \node[circ, label={right:$1$}] (s1131221) at ($(s113122) + (-150:\len)$) {};
     \node[circ, label={right:$3$}] (s1131222) at ($(s113122) + (-120:\len)$) {};
      \node[circ, label={right:$4$}] (s1131223) at ($(s113122) + (-60:\len)$) {};
      \node[circ, label={right:$12$}] (s1131224) at ($(s113122) + (-30:\len)$) {};

    \draw[arr] (s1) -- (s11);
    \draw[arr] (s1) -- (s12);
    \draw[arr] (s1) -- (s13);
    \draw[arr] (s11) -- (s111);
    \draw[arr] (s11) -- (s112);
    \draw[arr] (s11) -- (s113);
    \draw[arr] (s112) -- (s1121);
    \draw[arr] (s112) -- (s1122);
    \draw[arr] (s1122) -- (s11221);
    \draw[arr] (s113) -- (s1131);
    \draw[arr] (s1131) -- (s11311);
    \draw[arr] (s1131) -- (s11312);
    \draw[arr] (s11312) -- (s113121);
    \draw[arr] (s11312) -- (s113122);
    \draw[arr] (s113122) -- (s1131221);
    \draw[arr] (s113122) -- (s1131222);
    \draw[arr] (s113122) -- (s1131223);
    \draw[arr] (s113122) -- (s1131224);

\end{tikzpicture}}}%
        \hspace{1cm}%
        \longleftrightarrow%
        \hspace{1cm}%
        \vcenter{\hbox{\usetikzlibrary{positioning}
\tikzset{
  circ/.style={
    circle, draw, semithick, minimum size=1.5mm, inner sep=0pt, fill=black,
    }
}

\begin{tikzpicture}[
    arr/.style={semithick}
]
    \newcommand{\len}{1}

    \node[circ, label={right:$14$}] (s1) at (0,0) {};
    \node[circ, label={right:$2$}] (s11) at ($(s1) + (-135:\len)$) {};
    \node[circ, label={right:$5$}] (s12) at ($(s1) + (-90:\len)$) {};
    \node[circ, label={right:$6$}] (s13) at ($(s1) + (-45:\len)$) {};
    \node[circ, label={right:$8$}] (s111) at ($(s11) + (-135:\len)$) {};
    \node[circ, label={right:$11$}] (s112) at ($(s11) + (-90:\len)$) {};
    \node[circ, label={right:$13$}] (s113) at ($(s11) + (-40:\len)$) {};
    \node[circ, label={right:$7$}] (s1121) at ($(s112) + (-135:\len)$) {};
    \node[circ, label={right:$18$}] (s1122) at ($(s112) + (-45:\len)$) {};
    \node[circ, label={right:$16$}] (s11221) at ($(s1122) + (-135:\len)$) {};
    \node[circ, label={right:$9$}] (s1131) at ($(s113) + (-45:\len)$) {};
    \node[circ, label={right:$10$}] (s11311) at ($(s1131) + (-115:\len)$) {};
    \node[circ, fill=white, draw=black] (s11312) at ($(s1131) + (-45:\len)$) {};
     \node[circ, label={right:$15$}] (s113121) at ($(s11312) + (-135:\len)$) {};
     \node[circ, label={right:$17$}] (s113122) at ($(s11312) + (-90:\len)$) {};
     \node[circ, label={right:$1$}] (s1131221) at ($(s113122) + (-150:\len)$) {};
     \node[circ, label={right:$3$}] (s1131222) at ($(s113122) + (-120:\len)$) {};
      \node[circ, label={right:$4$}] (s1131223) at ($(s113122) + (-60:\len)$) {};
      \node[circ, label={right:$12$}] (s1131224) at ($(s113122) + (-30:\len)$) {};

    \draw[arr] (s1) -- (s11);
    \draw[arr] (s1) -- (s12);
    \draw[arr] (s1) -- (s13);
    \draw[arr] (s11) -- (s111);
    \draw[arr] (s11) -- (s112);
    \draw[arr] (s11) -- (s113);
    \draw[arr] (s112) -- (s1121);
    \draw[arr] (s112) -- (s1122);
    \draw[arr] (s1122) -- (s11221);
    \draw[arr] (s113) -- (s1131);
    \draw[arr] (s1131) -- (s11311);
    \draw[arr] (s1131) -- (s11312);
    \draw[arr] (s11312) -- (s113121);
    \draw[arr] (s11312) -- (s113122);
    \draw[arr] (s113122) -- (s1131221);
    \draw[arr] (s113122) -- (s1131222);
    \draw[arr] (s113122) -- (s1131223);
    \draw[arr] (s113122) -- (s1131224);

\end{tikzpicture}}}
        $} 
    \end{center}
  \caption{An example of the bijection between Cayley trees with $n$ labels and Cayley trees with $n-1$ labels and a virtual node. The virtual node in white is obtained by blanking out the node with largest label, $19$.}
    \label{fig:BijCayley-virtual}
\end{figure}
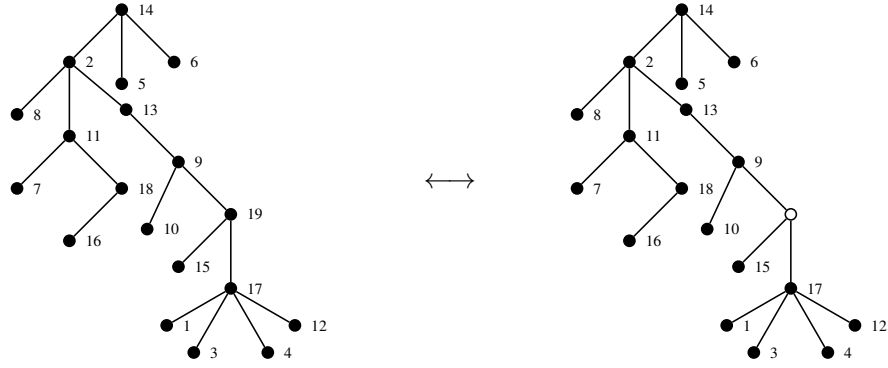
\begin{proof}There is a simple bijection between the trees of size $n-1$ in $\mathcal{R}^*_k$ and the trees of size $n$ in $\mathcal{R}_k$: this bijection  consists in giving the label $n$ to the virtual node of a tree of size $n-1$ from $\mathcal{R}^*_k$  and, conversely, in removing the label $n$ of a tree of size $n$ of $\mathcal{R}_k$ (see Figure~\ref{fig:BijCayley-virtual}). Since the virtual node is considered as a record and since it has no records in its sub-tree, the number of records does not change.
Then, 
$
\sum_{T \in \mathcal{R}^*_k} \frac{z^{|T|}}{|T|!}=\sum_{T' \in \mathcal{R}_k} \frac{z^{|T'|-1}}{(|T'|-1)!}=\mathcal{R}'_k(z).
$
\end{proof}

Let us now define an application $\phi$ mapping each element of $\mathcal{R}^*_k$ to a pair consisting after standardization in a Cayley tree and an element of $\mathcal{R}^*_{k-1}$:
Given a tree $T$ belonging to $\mathcal{R}^*_k$,  we mark the parent of the virtual node and cut the tree into two sub-trees at the level of the virtual node. 
We denote by  $S'$  the sub-tree of $T$ whose root is the virtual node, and  by  $S$ the sub-tree obtained from $T$ by stripping the sub-tree $S'$ and marking the node where it was attached (see Figure~\ref{fig:DecCayleyVirtual}, left-hand side). Then we transfer the label $\ell$ of the marked node of $S$ to the root of $S'$ to form a labeled tree $\bar S'$. Finally we shift all the labels $\ell_1<\ldots<\ell_p$ of $S$ that are greater than $\ell$ and remove the mark to get a tree $\bar S$: the marked node of $S$ gets label $\ell_i$ in $\bar S$, the node with label $\ell_i$  gets label $\ell_{i+1}$, and the node with maximal label $\ell_p$ in $S$ becomes a virtual (unlabeled) node in $\bar S$. 
\begin{p}
       The mapping $\phi$ that sends $T$ to $(\bar S,\bar S')$ is a bijection between
 trees of $\mathcal{R}^*_k$ of size $n$, and pairs of trees $(\bar S,\bar S')$ with total label set $\{1,\ldots,n\}$ such that the standardization of $\bar S$ belongs to $\mathcal{R}^*_k$ and the standardization of $\bar S'$ is a Cayley tree.

\end{p}
\begin{proof}
Observe that the tree $S$ has $k-1$ records since by definition of the class  $\mathcal{R}^*_k$ there are no records in the sub-tree starting at the virtual node, except for the virtual node itself, and the shift of labels preserves the records: in particular the virtual node of $\bar S$ remains a record since it takes the place of the maximum of $S$.
Clearly $S$ and $S'$ can be recovered from $\bar S$ and $\bar S'$ upon removing label $\ell$ at the root of $\bar S'$ and inserting it at the first larger label $\ell_1>\ell$ in $\bar S$ and shifting back the larger labels $\ell_1<\ldots<\ell_p$ until $\ell_p$ is put back on the virtual node of $\bar S$. Then $S$ and $S'$ are glued together to recover $T$. This proves that the mapping $\phi$ is injective. 

Conversely, given a pair of trees $(\bar S,\bar S')$ as in the proposition the inverse construction yields a tree $T$ with a virtual node and $k$ records that are not in the sub-tree of the virtual node, that is a tree of $\mathcal{R}^*_k$, and $\phi(T)=(\bar S,\bar S')$.
\end{proof}

 \begin{figure}[h!]
    \centering
\resizebox{1\textwidth}{!}{$%
  \vcenter{\hbox{\usetikzlibrary{positioning}
\tikzset{
  circ/.style={
    circle, draw, semithick, minimum size=1.5mm, inner sep=0pt, fill=black,
    }
}

\begin{tikzpicture}[
    arr/.style={semithick}
]
    \newcommand{\len}{1}

    \node[circ, label={right:$14$}] (s1) at (0,0) {};
    \node[circ, label={right:$2$}] (s11) at ($(s1) + (-135:\len)$) {};
    \node[circ, label={right:$5$}] (s12) at ($(s1) + (-90:\len)$) {};
    \node[circ, label={right:$6$}] (s13) at ($(s1) + (-45:\len)$) {};
    \node[circ, label={right:$8$}] (s111) at ($(s11) + (-135:\len)$) {};
    \node[circ, label={right:$11$}] (s112) at ($(s11) + (-90:\len)$) {};
    \node[circ, label={right:$13$}] (s113) at ($(s11) + (-40:\len)$) {};
    \node[circ, label={right:$7$}] (s1121) at ($(s112) + (-135:\len)$) {};
    \node[circ, label={right:$18$}] (s1122) at ($(s112) + (-45:\len)$) {};
    \node[circ, label={right:$16$}] (s11221) at ($(s1122) + (-135:\len)$) {};
    \node[circ, label={right:$9$}] (s1131) at ($(s113) + (-45:\len)$) {};
    \node[circ, label={right:$10$}] (s11311) at ($(s1131) + (-115:\len)$) {};
    \node[circ, fill=white, draw=black] (s11312) at ($(s1131) + (-45:\len)$) {};
     \node[circ, label={right:$15$}] (s113121) at ($(s11312) + (-135:\len)$) {};
     \node[circ, label={right:$17$}] (s113122) at ($(s11312) + (-90:\len)$) {};
     \node[circ, label={right:$1$}] (s1131221) at ($(s113122) + (-150:\len)$) {};
     \node[circ, label={right:$3$}] (s1131222) at ($(s113122) + (-120:\len)$) {};
      \node[circ, label={right:$4$}] (s1131223) at ($(s113122) + (-60:\len)$) {};
      \node[circ, label={right:$12$}] (s1131224) at ($(s113122) + (-30:\len)$) {};

    \draw[arr] (s1) -- (s11);
    \draw[arr] (s1) -- (s12);
    \draw[arr] (s1) -- (s13);
    \draw[arr] (s11) -- (s111);
    \draw[arr] (s11) -- (s112);
    \draw[arr] (s11) -- (s113);
    \draw[arr] (s112) -- (s1121);
    \draw[arr] (s112) -- (s1122);
    \draw[arr] (s1122) -- (s11221);
    \draw[arr] (s113) -- (s1131);
    \draw[arr] (s1131) -- (s11311);
    \draw[arr] (s1131) -- (s11312);
    \draw[arr] (s11312) -- (s113121);
    \draw[arr] (s11312) -- (s113122);
    \draw[arr] (s113122) -- (s1131221);
    \draw[arr] (s113122) -- (s1131222);
    \draw[arr] (s113122) -- (s1131223);
    \draw[arr] (s113122) -- (s1131224);

\end{tikzpicture}}}%
  \hspace{1cm}%
  \longleftrightarrow%
  \hspace{1cm}%
  \vcenter{\hbox{\usetikzlibrary{positioning}
\tikzset{
  circ/.style={
    circle, draw, semithick, minimum size=1.5mm, inner sep=0pt, fill=black,
    }
}

\begin{tikzpicture}[
    arr/.style={semithick}
]
    \newcommand{\len}{1}

    \node[circ, label={right:$14$}] (s1) at (0,0) {};
    \node[circ, label={right:$2$}] (s11) at ($(s1) + (-135:\len)$) {};
    \node[circ, label={right:$5$}] (s12) at ($(s1) + (-90:\len)$) {};
    \node[circ, label={right:$6$}] (s13) at ($(s1) + (-45:\len)$) {};
    \node[circ, label={right:$8$}] (s111) at ($(s11) + (-135:\len)$) {};
    \node[circ, label={right:$11$}] (s112) at ($(s11) + (-90:\len)$) {};
    \node[circ, label={right:$13$}] (s113) at ($(s11) + (-40:\len)$) {};
    \node[circ, label={right:$7$}] (s1121) at ($(s112) + (-135:\len)$) {};
    \node[circ, label={right:$18$}] (s1122) at ($(s112) + (-45:\len)$) {};
    \node[circ, label={right:$16$}] (s11221) at ($(s1122) + (-135:\len)$) {};
    \node[circ, fill=white, label={right:$9$}] (s1131) at ($(s113) + (-45:\len)$) {};
    \node[circ, label={right:$10$}] (s11311) at ($(s1131) + (-115:\len)$) {};
    \node[circ, fill=white, draw=blue] (s11312) at ($(s1131) + (-45:\len)$) {};
     \node[circ, fill=blue, draw=blue, label={right:$\textcolor{blue}{15}$}] (s113121) at ($(s11312) + (-135:\len)$) {};
     \node[circ, fill=blue, draw=blue,label={right:$\textcolor{blue}{17}$}] (s113122) at ($(s11312) + (-90:\len)$) {};
     \node[circ, fill=blue, draw=blue,label={right:$\textcolor{blue}{1}$}] (s1131221) at ($(s113122) + (-150:\len)$) {};
     \node[circ, fill=blue, draw=blue,label={right:$\textcolor{blue}{3}$}] (s1131222) at ($(s113122) + (-120:\len)$) {};
      \node[circ, fill=blue, draw=blue,label={right:$\textcolor{blue}{4}$}] (s1131223) at ($(s113122) + (-60:\len)$) {};
      \node[circ, fill=blue, draw=blue,label={right:$\textcolor{blue}{12}$}] (s1131224) at ($(s113122) + (-30:\len)$) {};

    \draw[arr] (s1) -- (s11);
    \draw[arr] (s1) -- (s12);
    \draw[arr] (s1) -- (s13);
    \draw[arr] (s11) -- (s111);
    \draw[arr] (s11) -- (s112);
    \draw[arr] (s11) -- (s113);
    \draw[arr] (s112) -- (s1121);
    \draw[arr] (s112) -- (s1122);
    \draw[arr] (s1122) -- (s11221);
    \draw[arr] (s113) -- (s1131);
    \draw[arr] (s1131) -- (s11311);
    
    \draw[arr, draw=blue] (s11312) -- (s113121);
    \draw[arr, draw=blue] (s11312) -- (s113122);
    \draw[arr, draw=blue] (s113122) -- (s1131221);
    \draw[arr, draw=blue] (s113122) -- (s1131222);
    \draw[arr, draw=blue] (s113122) -- (s1131223);
    \draw[arr, draw=blue] (s113122) -- (s1131224);

\end{tikzpicture}}}%
   \hspace{1cm}%
  \longleftrightarrow%
  \hspace{1cm}%
  \vcenter{\hbox{\usetikzlibrary{positioning}
\tikzset{
  circ/.style={
    circle, draw, semithick, minimum size=1.5mm, inner sep=0pt, fill=black,
    }
}

\begin{tikzpicture}[
    arr/.style={semithick}
]
    \newcommand{\len}{1}

    \node[circ, label={right:$16$}] (s1) at (0,0) {};
    \node[circ, label={right:$2$}] (s11) at ($(s1) + (-135:\len)$) {};
    \node[circ, label={right:$5$}] (s12) at ($(s1) + (-90:\len)$) {};
    \node[circ, label={right:$6$}] (s13) at ($(s1) + (-45:\len)$) {};
    \node[circ, label={right:$8$}] (s111) at ($(s11) + (-135:\len)$) {};
    \node[circ, label={right:$13$}] (s112) at ($(s11) + (-90:\len)$) {};
    \node[circ, label={right:$14$}] (s113) at ($(s11) + (-40:\len)$) {};
    \node[circ, label={right:$7$}] (s1121) at ($(s112) + (-135:\len)$) {};
    \node[circ, fill=white, draw=black] (s1122) at ($(s112) + (-45:\len)$) {};
    \node[circ, label={right:$18$}] (s11221) at ($(s1122) + (-135:\len)$) {};
    \node[circ, label={right:$10$}] (s1131) at ($(s113) + (-45:\len)$) {};
    \node[circ, label={right:$11$}] (s11311) at ($(s1131) + (-115:\len)$) {};
    \node[circ, fill=blue, draw=blue, label={right:\textcolor{blue}{$9$}}] (s11312) at ($(s1131) + (-45:\len)$) {};
     \node[circ, fill=blue, draw=blue, label={right:\textcolor{blue}{$15$}}] (s113121) at ($(s11312) + (-135:\len)$) {};
     \node[circ, fill=blue, draw=blue,, label={right:\textcolor{blue}{$17$}}] (s113122) at ($(s11312) + (-90:\len)$) {};
     \node[circ, fill=blue, draw=blue,, label={right:\textcolor{blue}{$1$}}] (s1131221) at ($(s113122) + (-150:\len)$) {};
     \node[circ, fill=blue, draw=blue, label={right:\textcolor{blue}{$3$}}] (s1131222) at ($(s113122) + (-120:\len)$) {};
      \node[circ, fill=blue, draw=blue, label={right:\textcolor{blue}{$4$}}] (s1131223) at ($(s113122) + (-60:\len)$) {};
      \node[circ, fill=blue, draw=blue, label={right:\textcolor{blue}{$12$}}] (s1131224) at ($(s113122) + (-30:\len)$) {};

    \draw[arr] (s1) -- (s11);
    \draw[arr] (s1) -- (s12);
    \draw[arr] (s1) -- (s13);
    \draw[arr] (s11) -- (s111);
    \draw[arr] (s11) -- (s112);
    \draw[arr] (s11) -- (s113);
    \draw[arr] (s112) -- (s1121);
    \draw[arr] (s112) -- (s1122);
    \draw[arr] (s1122) -- (s11221);
    \draw[arr] (s113) -- (s1131);
    \draw[arr] (s1131) -- (s11311);
    \draw[arr, draw=blue] (s11312) -- (s113121);
    \draw[arr, draw=blue] (s11312) -- (s113122);
    \draw[arr, draw=blue] (s113122) -- (s1131221);
    \draw[arr, draw=blue] (s113122) -- (s1131222);
    \draw[arr, draw=blue] (s113122) -- (s1131223);
    \draw[arr, draw=blue] (s113122) -- (s1131224);

\end{tikzpicture}}}%
  $}
  \caption{The decomposition $\phi$ of a Cayley tree with a virtual node. The trees obtained are $\bar S$, in black, and $\bar S'$, in blue. In the center there are the intermediary trees $S$ and $S'$.}
  \label{fig:DecCayleyVirtual}
\end{figure}
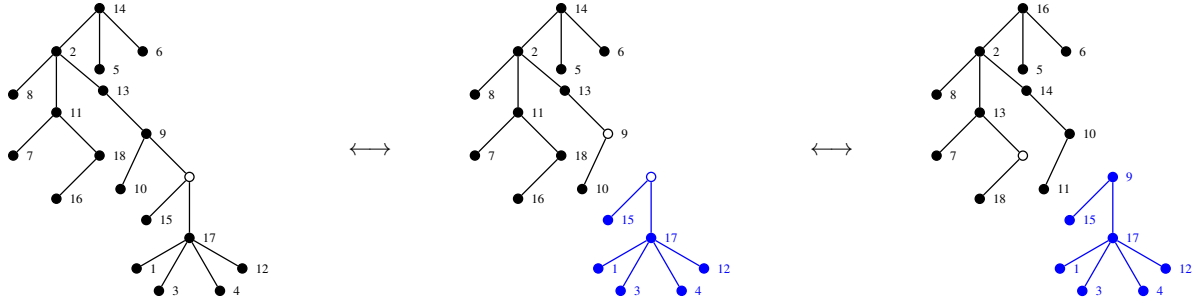

\begin{cor}[Proposition 3.4 of \cite{LRT-GF-records}]
\label{le:trees_k_records}
The generating function for   trees with precisely $k$ records can be expressed in terms of the Cayley tree function as
\begin{align*}
\mathcal{R}_k = 
\int_0^z  \frac{1}{s}
  \mathcal{T}^k(s)  \ ds.
  \end{align*}
\end{cor}

\section*{Acknowledgements}

We are grateful to Ira Gessel encouraging remarks and valuable suggestions, to Cyril Bandelier for engaging discussions, and to  Andrea Sportiello for his insightful remarks.

This work has been partially supported by Grants PID2020-117843GB-
I00 and PID2024-157173NB-I00 funded by MCIN/AEI/10.13039/501100011033 and by FEDER, UE. Author ST is supported by \href{10.3030/101070558}{FoQaCiA} which is funded by the European Union and NSERC.

\bibliographystyle{eptcsalpha}
\bibliography{references}

\end{document}